\documentclass{sig-alternate}

\usepackage{cite}

\newlength{\figtxtwid}
\long\gdef\boxitnew#1{\dimen200 = \hsize \advance\dimen200 by -7pt
\begingroup\vbox{\hrule \hbox to \hsize{\vrule\kern3pt
      \vbox{\hsize \dimen200 \kern3pt#1\kern3pt}\hfil \kern3pt\vrule}\hrule}\endgroup}

\newlength{\myfigwidth}
\newcommand{\boxtext}[1]{\setlength{\myfigwidth}{\columnwidth}\addtolength{\myfigwidth}{-1ex}\boxitnew{\vspace*{.1ex}\begin{minipage}{\myfigwidth}\setlength{\parindent}{0pt}#1\end{minipage}\\[.1ex]}}

\newtheorem{theorem}{Theorem}[section]
\newtheorem{lemma}[theorem]{Lemma}
\newtheorem{corollary}[theorem]{Corollary}
\newtheorem{definition}[theorem]{Definition}

\newcommand{\ceil}[1]{\lceil #1 \rceil}

\begin{document}

\title{Faster Information Dissemination in Dynamic Networks via Network Coding\vspace*{-0.5cm}}

\numberofauthors{2} 

\author{
\alignauthor
Bernhard Haeupler\\
       \affaddr{Massachusetts Institute of Technology}\\
       \affaddr{32 Vassar Street, 32-G622}\\
       \affaddr{Cambridge, MA 02139, USA}\\
       \email{haeupler@mit.edu}
       \vspace*{-0.3cm}
\alignauthor
David R. Karger\\
       \affaddr{Massachusetts Institute of Technology}\\
       \affaddr{32 Vassar Street, 32-G592}\\
       \affaddr{Cambridge, MA 02139, USA}\\
       \email{karger@mit.edu}
       \vspace*{-0.3cm}
}

%
% --- Author Metadata here ---
\conferenceinfo{PODC'11,} {June 6--8, 2011, San Jose, California, USA.}
\CopyrightYear{2011}
\crdata{978-1-4503-0719-2/11/06}
\clubpenalty=10000
\widowpenalty = 10000
% --- End of Author Metadata ---

\date{March 2011}
%\fi

\maketitle

\begin{abstract}
  We use \emph{network coding} to improve
  the speed of distributed computation in the dynamic network model of Kuhn, Lynch and Oshman 
  [STOC '10].  In this model an adversary adaptively chooses a new
  network topology in every round, making even 
  basic distributed computations challenging.  
  
  Kuhn et al. show that
  $n$ nodes, each starting with a $d$-bit \emph{token}, can broadcast
  them to all nodes in time $O(n^2)$ using $b$-bit
  messages, where $b \geq d+\log n$. Their algorithms take
  the natural approach of \emph{token forwarding}: in every round each
	node broadcasts some particular token it knows. They prove matching
	$\Omega(n^2)$ lower bounds for a natural class of token forwarding algorithms
	and an $\Omega(n \log n)$ lower bound that applies to all token-forwarding algorithms.

We use \emph{network coding}, transmitting random linear combinations of
tokens, to break both lower bounds.  Our algorithm's performance is
\emph{quadratic} in the message size $b$, broadcasting the $n$ tokens 
in roughly $\frac{d}{b^2} \cdot n^2$ rounds. For $b=d=\Theta(\log n)$
our algorithms use $O(n^2/\log n)$ rounds, breaking the first lower bound, while for larger
message sizes we obtain linear-time algorithms. We also consider networks that
change only every $T$ rounds, and achieve an additional factor $T^2$ speedup. This contrasts with related lower and upper bounds of Kuhn et al. implying that for natural token-forwarding algorithms a speedup of $T$, but not more, can be obtained. Lastly, we give a general
way to derandomize random linear network coding, that also leads to new deterministic information dissemination algorithms. 
\end{abstract}

\iffalse
We use network coding to improve the speed of distributed computation in the dynamic network model of Kuhn, Lynch and Oshman [STOC '10].  In this model an adversary adaptively chooses a new network topology in every round, making even basic distributed computations challenging.  

Kuhn et al. show that n nodes, each starting with a d-bit token, can broadcast them to all nodes in time O(n^2) using b-bit messages, where b > d + log n. Their algorithms take the natural approach of {token forwarding}: in every round each node broadcasts some particular token it knows. They prove matching Omega(n^2) lower bounds for a natural class of token forwarding algorithms and an Omega(n log n) lower bound that applies to all token-forwarding algorithms.

We use {network coding}, transmitting random linear combinations of tokens, to break both lower bounds.  Our algorithm's performance is {quadratic} in the message size b, broadcasting the n tokens in roughly d/b^2 * n^2 rounds. For b = d = O(log n) our algorithms use O(n^2/log n) rounds, breaking the first lower bound, while for larger message sizes we obtain linear-time algorithms. We also consider networks that change only every T rounds, and achieve an additional factor T^2 speedup. This contrasts with related lower and upper bounds of Kuhn et al. implying that for natural token-forwarding algorithms a speedup of T, but not more, can be obtained. Lastly, we give a general way to derandomize random linear network coding, that also leads to new deterministic information dissemination algorithms. 
\fi

\vspace*{-0.15cm}
%A category including the fourth, optional field follows...
\vspace*{-0.25cm}
\terms{\vspace*{-0.15cm}Algorithms, Performance, Theory}
\vspace*{-0.25cm}
\keywords{\vspace*{-0.15cm}dynamic networks, gossip, multicast, network coding}
\vspace*{-0.45cm}

\section{Introduction}

In this paper we demonstrate that \emph{network coding} can
significantly improve the efficiency of distributed computations in
dynamic networks. Network coding breaks with the classical paradigm of
routing atomic packets through a network and recognizes that
information can be mixed and coded together in ways other (physical)
quantities can not.  Network coding is a relatively recent discovery
that has already revolutionized information theory; it is
now a crucial tool in designing robust and efficient communication
protocols.  We believe network coding has potential for similar
impact in the distributed computing community.

We study the recently introduced \emph{dynamic network} model of
Kuhn et al. \cite{KLO}.  This model was designed
to capture the highly dynamic and non-converging nature of many modern
networks by allowing the network topology to change completely and
adaptively in every round subject to the constraint that the network
is always connected.  In each synchronized communication round, each node chooses a message
which is then broadcast to its neighbors for the round.
What makes this problem particularly challenging is that the broadcast is \emph{anonymous}, i.e., at the time a node
chooses its message it does not know who its receiving neighbors for the round will be.

An important problem in such dynamic networks is \emph{$k$-token   dissemination:}
there are $k$ tokens initially distributed
to some nodes, and the goal is to disseminate them to all nodes.

The most natural approach to solving token dissemination is \emph{token forwarding}:
in each round, each node chooses to broadcast one token it
knows.  Kuhn et al. \cite{KLO} show how to disseminate $k$ tokens in an $n$-node network in $O(nk)$ time
by flooding the $k$ tokens one by one in $O(n)$ rounds each. They also show how pipelining
can improve the running time of this approach to $O(\frac{nk}{T} + n)$ in slower-changing \emph{$T$-interval-connected} networks, in which for any
interval of $T$ rounds the links of some specific underlying spanning tree persist.

Kuhn et al. give evidence that this is the best one can do with token 
forwarding. For the natural class of \emph{knowledge-based} token forwarding algorithms, where each node's
messages depend only on the tokens it knows, they show a
matching $\Omega(\frac{nk}{T} + n)$ lower bound. They also give a more general
$\Omega(n \log k)$ lower bound that applies even if the algorithm is
operated under ``centralized control'' and mention in the conclusion the 
''hope to strengthen [this] and obtain an $\Omega(nk/T)$ general lower bound''.

Building on work of the first author~\cite{AnalyzingNC},
we show that these lower bounds cease to hold if one does not require that
tokens be broadcast individually.  
We use network coding, sending out random linear combinations of
tokens, to solve $k$-token dissemination of size-$O(\log n)$ tokens in $O(kn/\log n)$ time,
outperforming the $\Theta(kn)$ bound~\cite{KLO} for knowledge-based token forwarding
algorithms.  We also show that, perhaps counter-intuitively, larger
tokens can be disseminated faster: if the token size (and message
size) is $d$, network coding can disseminate $k$ tokens in $O(k(n\log
n)/d)$ time.  Thus, for tokens of size $n\log n$, we break the general
$\Omega(n\log k)$ bound on token-forwarding algorithms.   

We also consider networks that are \emph{$T$-stable}, changing only once every $T$ rounds.  Kuhn et al. show that token-forwarding can achieve a factor-$T$ speedup in this case, but that knowledge-based token-forwarding algorithms cannot do better.  In contrast, we show network coding can achieve a factor $T^2$ speedup.  

Finally, we show that linear network coding is not inherently randomized but
that the ideas and improvements carry over to (non-uniform) deterministic algorithms as well.
%\paragraph{Organization}
%The rest of the paper is organized as follows. We first 

\section{Our Results}

In this section we provide the formal statements of our main results. The model should be clear from the introduction but is also more formally described in Section~\ref{sec:model}.  

\subsection{The Role of Message Size}

Kuhn et al. assume throughout that the message size is equal to the token size.  For token-forwarding algorithms, this is quite reasonable.  For fixed token size, a larger message simply allows forwarding more tokens at once, which for all their results is equivalent to executing multiple rounds in parallel.  Thus, all their upper and lower bounds simply scale linearly with this message size.  

Once we move beyond token forwarding this equivalence breaks down.  Thus, we introduce a separate parameter, $b$, representing the size of a message.   We will see that network coding performance improves \emph{quadratically} with the message size.  Somewhat surprisingly, this means that when the message size is equal to the token size, \emph{larger tokens can be disseminated faster}.  

Explicitly modeling $b$ also allows us to bridge an important gap between the distributed computing and network coding communities.  In distributed computing we often focus on size-$O(\log n)$ message-sizes.  But in practice, most communication protocols impose a minimum message size in the thousands or tens of thousands of bits.  We should therefore try to take advantage of the possibility of tokens being much smaller than the message size; with network coding we can.  At the other end, the network coding community generally assumes messages are so large that overheads associated with network coding can be ignored.  Our work accounts for the hidden cost of these overheads, which can be significant when messages are smaller.  
In summary, explicitly modeling $b$ lets us span the range of assumptions from distributed computing's tiny messages to network coding's huge ones.  We discuss this in more detail in Section~\ref{sec:related}.

\subsection{Token Forwarding Algorithms}

For comparison we first recall the upper- and lower-bound results of \cite{KLO}:

\begin{theorem}\cite{KLO}\label{thm:klo}
There is a deterministic knowledge-based token forwarding algorithm that solves the $k$-token dissemination problem in a $T$-stable dynamic network in $O(\frac{1}{T} \cdot \frac{nkd}{b} + n)$ rounds using messages of size $b$ for tokens of size $d$. This is tight, i.e., for any $T$, any (even randomized) knowledge-based token forwarding algorithm takes at least $\Omega(\frac{1}{T} \cdot \frac{nkd}{b} + n)$ rounds in the worst case. 
\end{theorem}

This is not a verbatim restatement. Indeed, Kuhn et al \cite{KLO} prove this theorem for the related but stronger stability measure of $T$-interval connectivity. Furthermore, except for the abstract, they only describe the case of small tokens and assume that the messages size is equal to the size of the tokens, i.e., $b = d = \log n$. Lastly, for most of the paper they assume that $k=n$ or that each node starts with exactly one token. It is easy to verify that the lower bound from \cite{KLO} continues to hold for our weaker $T$-stability model and that the algorithms also directly extend to the stated theorem: E.g., to achieve a running time of $\frac{nkd}{b}$ for $T=1$ the nodes repeatedly flood $\frac{b}{d}$ tokens per $O(n)$ rounds instead of one. 

Their second lower bound applies to deterministic centralized algorithms and shows that even if one allows such unrestricted,  coordination between nodes a linear time algorithm is not achievable (in contrast to static graphs):

\begin{theorem}\cite{KLO}\label{thm:kloglobal}
For $b=d$ any deterministic centralized token forwarding algorithm that solves the $k$-token dissemination problem in a dynamic network takes $\Omega(n \log k)$ rounds in the worst case. 
\end{theorem}

\subsection{Network Coding}

Even though the token dissemination problem is about delivering complete tokens, one can benefit from not treating the information as a physical quantity that needs to be routed through the network. We do this by providing faster (knowledge-based) algorithms for the $k$-token dissemination problem based on network coding.  The lower bound in Theorem~\ref{thm:klo} pertains even if one allows the algorithms to chop up tokens into single bits and route those bits independently through the network -- including concatenating bits of different tokens within one message. This shows that true (network) coding is required. 

Our algorithms use \emph{random linear network coding}, the arguably simplest form of network coding, in which messages are random linear combinations of tokens. Independent of network dynamics, nodes in our algorithm always choose a uniformly random linear combination of all received messages  and can therefore also be considered knowledge-based.

Our first theorem shows that one can solve $k$-token dissemination  roughly a factor of $b$ faster than the lower bound for knowledge-based token forwarding algorithms:

\begin{theorem}\label{thm:bandwidth}
There is a randomized network coding algorithm that solves the $k$-token dissemination problem in a dynamic network with $n$ nodes in  
$$O(\min\{ \frac{1}{b} \cdot \frac{nkd}{b} + nb, \frac{\log n}{b} \cdot \frac{nkd}{b} + n \log n\})$$
rounds with high probability.
\end{theorem}

This means that the efficiency of token-dissemination increases at least \emph{quadratically} with the message size, instead of the more intuitive linear increase given by Theorem~\ref{thm:klo}. A similar result is true for the advantages coming from more stable networks. Theorem~\ref{thm:klo} implies that $T$-stability (or even $T$-interval connectivity) allows for a speed up of $T$. 
%% For $k=n$ and $b=d$ this nicely interpolates the $O(n^2)$ rounds
%% needed for dynamic graphs and the obvious linear time pipelined
%% flooding solution for static graphs. 
Our next theorem shows that with network coding the speedup of a more stable network improves to $T^2$. For most parameter values, this improvement can be combined with the speed-up from larger message sizes. The next theorem implies an at least $\frac{log^2 n}{bT^2}$ speed up over the $O(\frac{nkd}{b})$ rounds for most settings of the parameters $b,d,k$ and $T$. This is a factor of $\frac{log^2 n}{bT}$ faster than the lower bound for knowledge-based token-forwarding algorithms: 

\begin{theorem}\label{thm:stable}
There is a randomized network coding algorithm that solves the $k$-token dissemination problem in a $T$-stable dynamic network with $n$ nodes in
$$\begin{array}{llll}
O(1)\cdot\min\Big\{\vspace{0.2cm}
  &\frac{\log n}{bT^2} \ \cdot \ \frac{nkd}{b} \ \  &+ \ \ nbT^2\log n&\\ \vspace{0.2cm}
  &\frac{\log^2 n}{bT^2} \ \cdot \ \frac{nkd}{b} \ \  &+ \ \ nT\log^2 n&\\ 
  &\frac{\log^2 n}{bT^2} \ \cdot \ \ n^2 \ \  &+ \ \ n\log n& \Big\}
\end{array}
$$
%$$ O(\min\{\frac{log^2 n}{bT^2} \cdot \frac{nkd}{b} + nbT^2\log^2 n, \frac{log^3 n}{bT^2} \cdot \frac{nkd}{b} + nT\log n, \frac{\log^2 n}{bT^2} \cdot n^2 + n \log^2 n\} )$$
rounds with high probability.
\end{theorem}

All these algorithms are based on random linear network coding which seems to be inherently dependent on randomization. We show that this is not true. We give tight trade-offs between the adaptiveness of the adversary and the required coefficient size/overhead.  For derandomization we must pay higher (quadratic) coefficient overhead, but we can still outperform token-forwarding algorithms.
These arguments apply quite generally to the network coding framework in~\cite{AnalyzingNC} and are interesting on their own. We defer the description of these results to Section~\ref{sec:det} and mention here only the implications for the $k$-dissemination problem: 

%Even with a deterministic knowledge-based algorithm it is possible to improve over the lower bound from Theorem~\ref{thm:klo} by a factor of $\frac{MIS(n)^2}{\sqrt{bT}})$, where $MIS(n)$ is the time needed to compute a maximal independent set on a graph of $n$ nodes. Given that the fastest known deterministic distributed MIS-algorithm (\cite{panconesi1992improved}) has a running time of  $2^{O(\sqrt{\log n})}$, we obtain the following theorem:

\begin{theorem}\label{thm:stabledeterministic}
There is a deterministic network coding algorithm that solves the $k$-token dissemination problem in a $T$-stable dynamic network with $n$ nodes in  
$$O(\frac{1}{\sqrt{bT}} \cdot n \cdot \min\{k,\frac{n}{T}\} + n) \cdot 2^{O(\sqrt{\log n})}$$
rounds.
\end{theorem}

For completeness we also describe what our findings imply for centralized algorithms\footnote{A centralized algorithm can globally coordinate nodes. Formally we define centralized algorithms as ``distributed'' algorithms that furthermore provide each node with knowledge about past topologies, the initial token distribution (without getting to know the tokens itself) and a source of shared randomness in case of a randomized algorithm. It is easy to verify that this extends the definition given in \cite{KLO} for centralized token-forwarding algorithms to general algorithms and problems.}:

\begin{corollary}\label{cor:global}
There is a randomized centralized network coding algorithm that solves the $k$-token dissemination problem in a $T$-stable dynamic network with $n$ nodes in  
order-optimal $\Theta(n)$ time with probability $1-2^{-n}$ and a deterministic centralized network coding algorithm that runs in $O(\frac{\log n}{bT} \cdot n \cdot \min\{k,\frac{n}{T} \} + n)$ rounds.
\end{corollary}

To help interpret these general results we present a few interesting value instantiations:
\begin{itemize}
	\item Even for $b=d=\log n$ and $k=n$, which is an important case because of its connection to counting the number of nodes in a network~\cite{KLO}, the $n^2/\log n$ rounds needed by the network coding algorithm is a $\Theta(\log n)$-factor faster than any knowledge-based token forwarding algorithm can be. 
	\item For the counting problem with larger message sizes, i.e., $d = \log n$ and  $k=n$, Theorem~\ref{thm:bandwidth} implies that a message-size of $b = \sqrt{n} \log n$ suffices to obtain an optimal linear-time randomized algorithm. For $b = n^{2/3} \log n$ this can be made deterministic. In contrast, the best known token-forwarding algorithm needs $b= n \log n$ (see Proposition 3.2 of \cite{KLO}) which is tight for knowledge-based token forwarding algorithms. 
	\item The situation is similar if one considers the question of how stable a graph needs to be to allow near-linear $n^{1+o(1)}$ time algorithms for the $n$-token dissemination problem. Theorems \ref{thm:stable} and \ref{thm:stabledeterministic} show that $T = \Omega(\sqrt{n})$ suffices for randomized algorithms and $T = \Omega(n^{2/3})$ for deterministic algorithms. This means that $\sqrt{n}$ (resp. $n^{1/3}$) adversarial topology changes can be tolerated with network coding. In contrast any knowledge-based token-forwarding algorithm requires the graph to be essentially static, i.e., $T = \Omega(n^{1-o(1)})$.
	\item 
For the case that messages are of the size of a token, i.e, $b = d$, the weaker but quite general lower bound for Theorem~\ref{thm:kloglobal} rules out any linear time token forwarding algorithm even if a deterministic centralized algorithm is used. In contrast to this there are linear time network coding algorithms that are: \vspace{-0.1cm}
	\begin{itemize}
		\item randomized and centralized 
		\item deterministic and centralized\\(for message and token sizes $\geq n \log n$)
    \item randomized and knowledge-based\\(for message and token sizes $\geq n \log n$)
		\item deterministic and knowledge-based\\(for message and token sizes $\geq n^2 \log n$)
	\end{itemize}
\end{itemize}

\section{Related Work}
\label{sec:related}
While traditional distributed algorithms research has focused on computation in 
static networks, the analysis of dynamic network topologies has gained importance both in practice and theory. 
Kuhn et al.~\cite{KLO} offer an extensive review of this literature. 

Next to \cite{KLO} the line of research most relevant to this work is network coding for gossip problems~\cite{informationdissemination05,debmed06transinf,borokhovich2010tight,mosk2006information}
and most specifically work by Haeupler~\cite{AnalyzingNC}. Since its introduction~\cite{ahlswede2000network,li2003linear}
network coding has revolutionized the understanding
of information flow in networks and found many practical applications (see, e.g., the books \cite{yeung2008information,ho2008network}).

Random linear network coding and its distributed implementation considered in this paper
were introduced by Ho et al.~\cite{ho2006random} and shown to achieve capacity for multicast.
Its performance for the distributed $n$-token dissemination problem has been intensively studied
in combination with gossip algorithms under the name of algebraic gossip or rumor spreading. 
The first such analysis~\cite{informationdissemination05,debmed06transinf} studied  the 
performance of algebraic gossip in the random phone call model, i.e., the complete graph in which 
each nodes sends a message to a random neighbor in each round.  Follow-on work~\cite{borokhovich2010tight,mosk2006information,avin2011podc,AnalyzingNC} has analyzed 
the distributed network coding gossip algorithm on general static networks. Haeupler~\cite{AnalyzingNC}
gives a very simple analysis technique (reviewed in Section~\ref{sec:NC})  that can be used to 
show order optimal stopping times in practically all communication models. Most interestingly
this holds true even if, as studied here and in \cite{KLO}, a fully adaptive adversary changes the topology in every
round. In the setting considered here this would imply an optimal $O(n)$ linear time algorithm for the
$n$-token dissemination problem. Unfortunately, these prior results do not directly
apply for two subtle but important reasons: 

First, \cite{AnalyzingNC}, as well as all prior work on algebraic gossip, assumes that
the additive overhead of the network coding header, which is linear in the number
of coded packets, is negligible compared to the size of a packet. This assumption is backed up by many practical implementations in which
this overhead is indeed less than one percent.  But a rigorous theoretical treatment, like that of~\cite{KLO}, must account for this overhead which may be significant if message-sizes are small.

Secondly, in all prior literature including \cite{AnalyzingNC}, it is also assumed that tokens are uniquely 
numbered/indexed and that this index is known to any node that starts with a token. This is needed to allow nodes
to specify in the coding header which packets are coded together in a message. In this paper such an assumption
would be unacceptable. For example, for the task of counting the number of nodes in a dynamic network \cite{KLO}
having the IDs consecutively indexed would essentially amount to assuming that a solution to the counting
problem is already part of the input. 

In this paper we address both points explicitly. Accounting for the coding overhead leads to interesting trade-offs
and poses new algorithmic challenges like the need for \emph{gathering} many tokens in one node so that
they can be grouped together to a smaller number of larger ``meta-tokens'' that require fewer coefficients.  To this end we consider
intermediate message sizes $b$ that can range between logarithmic size~\cite{KLO} to 
(super)linear size~\cite{informationdissemination05,debmed06transinf,borokhovich2010tight,mosk2006information,avin2011podc,AnalyzingNC}.
We furthermore do not assume any token indexing or other extra coordination between nodes but show how to bootstrap
the token dissemination algorithms to find such an indexing.

\section{Problem Description}\label{sec:model}

Throughout this paper we work in the dynamic network model of Kuhn et
al. \cite{KLO}.  The following section gives a detailed description of
the model and of the token dissemination problem.

\subsection{The Dynamic Network Model}

A \emph{dynamic network} consists of $n$ nodes with unique identifiers
(UIDs) of size $O(\log n)$ and we assume that the number of nodes is known (up to a factor of 2)
to all nodes.
The network operates in synchronized \emph{rounds}.
During each round $t$ the network's connectivity is defined by a
connected undirected graph $G(t)$ chosen by an adversary. The nodes
communicate via \emph{anonymous broadcast}: At the beginning of a
round each node chooses an $O(b)$-bit message, where $b \geq \log n$, without knowing to which
nodes it is connected in the round. After the messages and the
network $G(t)$ is fixed each node receives all messages chosen by its
neighbors in $G(t)$. The model does not restrict local computations
done by nodes.

We present deterministic and randomized algorithms. In the case of
randomization one must carefully specify how the adversary is allowed
to adapt to algorithmic actions.  We cover several models in the full paper but here we assume an \emph{adaptive adversary}:
in each round the adversary chooses the network
topology based on all \emph{past actions} (and the current state) of
the nodes. Following this the nodes then choose random messages (still
without knowing their neighbors).\\

\vspace{-0.15cm}
{\bfseries Remarks:}
\vspace{-0.15cm}
\begin{itemize}
	\item For randomized algorithms the assumption of $O(\log n)$
          size UIDs is without loss of generality since they can
          be generated randomly with a high probability of success.\vspace{-0.15cm}
	\item In the case of $n$-token dissemination the assumption that all
        nodes know $n$ is without loss of generality: If $n$ is
        unknown one can start with guessing an upper bound $n=2$,
        count the number of node IDs using $n$-token dissemination
        and repeatedly double the estimate an restart when a failure
        is detected. This use of the $n$-token dissemination prevents a termination
        with a too small estimate. Since the running times only depend
        (at least linearly) on the size of the estimate, all rounds spend on computations with 
        too low estimates are dominated by a geometric sum and increase the 
        overall complexity at most by a factor of two. A similar 
        argument was given in \cite{KLO}. We defer more details to the full paper. 
\end{itemize}

\subsection{The $k$-Token Dissemination Problem}

In this section we describe the \emph{$k$-token dissemination problem}
\cite{KLO}.  In this problem, $k \leq n$
\emph{tokens} of $d \leq b$ bits are located in the network and the goal is
for all nodes to become aware of the union of the tokens and then
terminate.  We assume that the $k$ tokens are chosen and distributed
to the nodes by the adversary before the first round.

Kuhn et al. observe that $k$-token dissemination seems intimately
connected to the problem of counting the number of nodes in a network
and to simpler problems like consensus.  In fact $k$-dissemination is
``universal'' as any function of the $k$ tokens can be computed by
distributing them to all nodes and the letting each node compute the
function locally.

We consider only \emph{Las Vegas} algorithms that are guaranteed to terminate with all tokens disseminated. We will bound the expected
number of rounds until all nodes terminate. All stopping times actually
hold with high probability.

Our algorithms for $k$-token dissemination solve several natural
subproblems as subroutines:\vspace{-0.1cm}
\begin{description}
	\item[gathering:] nodes need to collect tokens such that a single 
node or a small collection of nodes knows about a specified number of tokens. \vspace{-0.1cm}
	\item[$k$-indexing:] $k$ tokens must be selected
and a distinct \emph{index} in the range $1,\ldots,k$ assigned to each.\vspace{-0.1cm}
	\item[$k$-indexed-broadcasting:] $k$ 
          tokens with distinct indices $1,\ldots,k$ must be
          distributed to all nodes 
\end{description}

\section{(Analyzing) Network Coding}\label{sec:NC}

%In this section we introduce the ideas of network coding and review
%the projection analysis technique given in \cite{AnalyzingNC}. 

\subsection{Random Linear Network Coding}

Instead of sending the $d$-bit tokens as atomic
entities, network coding interprets these tokens as vectors over a
finite field and sends out random linear combinations of the
vectors. Formally, the algorithm chooses a prime $q$ as a field
size and represents the tokens as $d' = \ceil{d / \lg q}$-dimensional vectors over $F_q$.  For most of this paper one can choose
$q=2$, i.e., take the natural token representation as a bit sequence
of length $d' = d$ and replace linear combinations by XORs.

Let $t_1,\ldots,t_k \in F_q^{d'}$ be $k$ indexed tokens.  We
concatenate the $i^{th}$ basis vector $e_i$ of $F_q^{k}$ to $t_i$ to
produce a $k + d'$-dimensional vector $v_i$. Each
node that initially knows $t_i$ ``receives'' this vector $v_i$ before the first round.
Notice that if a node knows the \emph{subspace $S$ spanned by the $v_i$}, e.g, in
the form of any basis of $S$, it can use Gaussian elimination to 
reconstruct the $v_i$, and thus the original tokens. Thus, we solve $k$-indexed-broadcast by 
delivering to every node a set of vectors that span $S$. The algorithm is straightforward:
At each round, any node computes a \emph{random linear combination} of any 
vectors received so far (if any) and broadcasts this as a message to its (unknown) neighbors.
Note that the message only depends on the current knowledge of the tokens, i.e., the subspace
spanned by the received vectors. This natural property was called knowledge-based in \cite{KLO}.

\subsection{Advantages of Network Coding}\label{sec:advantageNC}
\label{sec:coding-vs-forwarding}

To contrast network coding  with token forwarding,
consider the simplified setting in which a node $A$ knows about all $k$
tokens while another node $B$ knows all but one token. If $A$ does not know which token $B$ is missing then, in a worst-case deterministic setting, $k$ rounds
of token forwarding are required. Randomized strategies can improve the expected number of 
rounds only to $k/2$. A better strategy is to send an XOR of all tokens: with this one piece of information $B$
can reconstruct the missing token. 

Similar situations arise frequently in the
end phase of token forwarding algorithms. Here  
most nodes already know most of the tokens but, because of the changing topology,
do not know which few tokens are not shared with their unknown neighbors of this round.
 Most token forwarding steps are therefore wasted. Network coding circumvents this problem, making it
\emph{highly probable} that \emph{every} communication will carry new information.

\subsection{The Network Coding Analysis}\label{sec:simpleNC}

In this section we review the simple projection analysis technique that was introduced previously~\cite{AnalyzingNC}. It shows that the full ``span'' of the message vectors ultimately spreads everywhere by tracking the projection of the received space in each direction separately. As argued above, a node $u$ can recover a token $t_i$ if and only if the first $k$-components of the vectors received by $u$ span the $i^{th}$ unit vector of $F_q^k$. For the analysis we will thus solely concentrate on the first $k$ coordinates of the vectors sent around. We track these projections using the following definition:

\begin{definition}\label{def:sensing}
A node $u$ \emph{senses} a coefficient vector $\vec \mu \in F_q^{k}$ if it
has received a message with a coefficient vector $\vec \mu'$ that is not
orthogonal to $\vec \mu$, i.e., $\vec  \mu' \cdot \vec  \mu \neq 0$. 
\end{definition} 

\begin{lemma}\label{lem:sensing-spreads}
Suppose a node $u$ senses a vector $\vec \mu$ and generates a new
message. Any recipient of this message will then sense
$\vec \mu$ with probability at least $1 - 1/q$. 
%Furthermore 
%if a node $w$ senses all vectors in $F_q^k$ then it is able
%to decode all tokens. 
\end{lemma}
\begin{proof}
This lemma simply states that a random linear combination of vectors
$\vec \mu'_j$ that are not all perpendicular to $\vec \mu$ is unlikely to be
perpendicular to $\vec \mu$.  Let $r_j$ be the random coefficient for $\vec \mu'_j$.
Then $(\sum r_j \vec \mu'_j) \cdot \mu = \sum r_j (\vec \mu'_j \cdot \mu)$.  Suppose 
without loss of generality that $\vec \mu'_0 \cdot \vec \mu \ne 0$.
Conditioned on all other values $r_j$, exactly one value of $r_0$
 will make the sum vanish. This value is taken with probability $1/q$.  
\end{proof}

Lemma~\ref{lem:sensing-spreads} shows that any node sensing any $\vec
\mu$ will pass that sense to its neighbors with constant probability.
Note furthermore that sensing is monotone and that unless all nodes
can already sense $\vec \mu$ the adversary must connect the
nodes that sense $\vec \mu$ to those that do not. This shows that in
each round the number of nodes that sense a vector $\vec \mu$
increases by a constant in expectation. A simple Chernoff bound
shows further that the probability that after $O(n + k)$ steps not all
nodes sense $\vec \mu$ is at most $q^{-\Omega(n+k)}$.
%why does k appear here? we are talking about only one vector. --- Because we run for n+K steps, we furthermore need the k here to apply the union bound}  
We now apply a union
bound: there are $q^{k}$ distinct vectors in $F_q^{k}$, and each
fails to be sensed by all nodes with probability $q^{-\Omega(n+k)}$. This
shows that all vectors in $F_q^{k}$ are sensed with high probability
implying that all nodes are able to decode all tokens.  The following
lemma is immediate.

\begin{lemma}\label{lem:simpleNC}
The network coding algorithm with $q \geq 2$ solves the $k$-indexed-broadcast problem
in an always connected dynamic network with probability at least $1 - q^{-n}$ in time $O(n + k)$.
It uses  messages of size $k \lg q + d$ where $d$ is the size of a token.
\end{lemma}

\section{Derandomizing Random Linear Network Coding}\label{sec:det}

The description of network coding above might suggest that the distributed
random linear network coding approach is inherently randomized.
We give the novel result that this is 
not the case. Instead of providing a deterministic algorithm directly we 
first prove that even an \emph{omniscient adversary}, which knows
knows all randomness in advance,  
cannot prevent the fast mixing of the network coding algorithm
if the field size is chosen large enough:

\begin{theorem}\label{thm:NComniscient}
The network coding algorithm with $q = n^{\Omega(k)}$ solves the $k$-indexed-broadcast problem in
an always connected dynamic network against an omniscient adversary with probability at least $1 - q^{-n}$ in time $O(n + k)$.
It uses messages of size $k^2 \log n + d$ where $d$ is the size of a token.
\end{theorem}
\begin{proof}(Sketch)
The proof of this result is nontrivial.  The obvious approach, of
taking a union bound over all possible adversarial strategies
expressed as a ``connectivity schedule,'' fails because there are too
many of them.  Instead, we carefully map each such schedule to a small
set of canonical ``witnesses'' that describe only the flow of new
information from node to node; there are few enough of these witnesses
that a union bound can be applied.

  We specify a compact witness by specifying, at each time step, \emph{which
    nodes learn something new} (in other words, receive a vector not
  already in the span of their received messages) and \emph{which nodes they learn it
    from}.\footnote{There may be some ambiguity about which received
    vectors are ``new'' if they are not linearly
    independent.  To remove this ambiguity, consider the vectors to
    arrive one at a time in some arbitrary order, and include the
    prior-arrived vectors of the round while evaluating newness.}
  Given all the random choices for the coefficients, this information
  suffices to inductively reconstruct the complete learning history (but not the complete topology sequence):
  By induction, we will know which subspace is spanned by each
  node at a given time step and, from the coefficient choices, we
  will know what vector it broadcasts. Given this, if we know which
  nodes learn something new from which nodes, we will know 
  what vectors each received and can thus infer what their subspace
  will be in the next round.  

  The key benefit of this representation is that it is small.  Note
  that nodes are learning a $k$-dimensional subspace, and that each
  time a node learns something new, the dimension of its subspace
  increases.  Thus, each node can have at most $k$ ``learning
  events''.  We specify the witness by specifying, for each node,
  the $k$ times and senders triggering such an events.  This requires
  $O(k\log n)$ bits per node for a total of $O(nk\log n)$ bits to
  specify a witness, meaning the number of witnesses is $\exp(nk\log
  n)$. With a failure probability of at most $q^{-n}$ and the given
  choice of $q$, this is sufficiently small for the union bound to apply;
  details will appear in the full paper.
\end{proof}

%The following derandomization results can now be proven in a straightforward manner. 
The proof of Theorem~\ref{thm:NComniscient} can be extended to a randomized existence proof for a
matrix that contains a sequence of pseudo-random choices for every possible ID; such that, no
matter how the adversary assigns the IDs and decides on the network dynamics, if all nodes choose their 
coding coefficients according to their sequence, all vectors always spread. By giving such a matrix as a (non-uniform) advice or by computing the, e.g., lexicographically first such matrix at every node, the next corollary follows. We defer the details to the full paper. 

\begin{corollary}\label{col:NCdeterministic} \ 
There are uniform and non-uniform deterministic algorithms that solve the $k$-indexed-broadcast problem 
in an always connected dynamic network in time $O(n + k)$ using messages of size
$k^2 \log n + d$ where $d$ is the size of a token. The uniform deterministic algorithm
performs a super-polynomial time local computation before sending the first message.  
\end{corollary}

\section{Token Dissemination with\\ Network Coding}\label{sec:generalization}

We now bridge the gap from index broadcast to token dissemination.  We begin with a simple result.
Combining the results from \cite{KLO} and Lemma~\ref{lem:simpleNC} yields the following corollary:

\begin{corollary}\label{cor:straightforward}
There is a randomized network coding algorithm that solves $k$-token dissemination in 
$O(\frac{nk \log n}{b}) = O(\frac{\log n}{d} \cdot \frac{nkd}{b})$ rounds with high probability.
\end{corollary}
\begin{proof}
  All nodes can generate $O(\log n)$-size unique IDs for their own
  tokens by concatenating a sequence number to the node ID.  Now all
  nodes flood the network repeatedly announcing the smallest
  $\Omega(b/\log n)$ tokens they have heard about.  After $n$ rounds all
  nodes will know these token IDs and can give them consistent distinct indices by sorting them.
  The corresponding $\Omega(b/\log n)$ tokens can then be broadcast
  to all nodes in $O(n)$ time using network-coded indexed broadcast.
  This needs to be repeated $k \frac{\log n}{b}$ times, leading to the claimed time bound.
\end{proof}

Unfortunately, this is only a $\frac{\log n}{d}$ factor faster than the bound for token forwarding algorithms from Theorem~\ref{thm:klo}. Thus no improvements are achieved for $d=O(\log n)$-size tokens, even for large message sizes. This is unsurprising as the algorithm uses flooding to solve the problem of disseminating the $b/\log n$ smallest token identifiers for indexing---a $k=(b/\log n)$-token dissemination problem with the identifiers treated as tokens of size $\Omega(\log n)$.  Thus if the tokens themselves are of logarithmic size relying on flooding as an indexing subroutine cannot lead to any improvement. We also note that, if $d \ll b$, the efficiency of the network coding messages is severely handicapped: The $O(b)$-size coefficient overhead takes up nearly all the space while the coded tokens only have size $d$. Thus in principle one could broadcast tokens that are a factor of $\frac{b}{d}$ larger.

We solve both problems by \emph{gathering} many tokens to one (or a small number of) nodes.  If all tokens are at one node, they can all trivially be assigned distinct indices.  Then, they can be grouped into blocks of $b/2d$ tokens, each of total size $b/2$, and network coding can be used to disseminate $b/2$ of these blocks simultaneously.  We need an additional $b/2$ space to hold the extra $b/2$ dimensions needed to ``untangle'' the coded messages, but these too fit in the size-$b$ messages.  In the discussion below, we will ignore the factors of 2 mentioned here.

%In the following we concentrate solely on online-adaptive adversaries. We
%assume that $k$ tokens of size $d$ are distributed over the network.

We have two gathering-based algorithms,  one that works well as long as $b \leq k^{1/3}$ and
one that works for larger message sizes. Both are based on the following simple random
token forwarding algorithm:\\[-0.15cm]

\boxtext{
\begin{tabbing}
else\= else\= else\= \kill
  Algorithm {\tt random-forward}\\
\\[-0.15cm]
  {\bf repeat} $O(n)$ times\\
  \>each node forwards $b/d$ tokens\\
  \>\ \ \ chosen randomly from those it knows\\
  \\
  Identify a node with the maximum token count\\
  \ \ \ (using $O(n)$ rounds of flooding)
  %{\bf repeat} $O(n)$ times\\
  %\>each nodes forwards the highest token count heard so far\\
  %\> together with the node UIDs (breaking ties arbitrarily)\\
\end{tabbing}
}
% \begin{figure}[h]
% \label{fig:random-forward}
% \caption{A Random Forward Algorithm}
% \end{figure}

\begin{lemma}\label{lem:randomforward}
  If initially there are $k$ tokens in the network then, after {\tt random-forward},
  the identified node knows with high probability either all or at least $M=\sqrt{\frac{bk}{d}}$ tokens.
\end{lemma}

\begin{proof} (Sketch)
  While there are less than $M$ tokens at any node, a node
  choosing $b/d$ random tokens to transmit will choose any
  \emph{particular} token with probability at least $b/dM$.  Since at least one
  node that knows the token is connected to one that do not, this
  implies that a token ``spreads'' to at least one new node each round
  with probability at least $b/dM$.  Thus after $n$ rounds each token
  is at $\Omega(bn/dM)$ nodes with high probability.  This applies to
  each token so there are $kbn/dM$ copies of tokens in the network.
  It follows that some node has at least $kb/dM$ tokens. A contradiction would
  arise unless $M > kb/dM$; the result follows. We defer the details to the full paper.
\end{proof}

This lemma has a nice interpretation, if one looks how tokens spread over time.  At first, the protocol is extremely efficient, but as more and more tokens become known to the nodes, there are ever more wasted broadcasts.  Spreading
all tokens in this way requires in expectation $O(nkd/b)$ rounds, because the wasted broadcasts occurring for the last half 
of the tokens dominate (see also Section~\ref{sec:coding-vs-forwarding}). Note that this is exactly the time bound for the 
flooding-based algorithms of Theorem~\ref{thm:klo}. Our first algorithm uses the efficient start phase of
{\tt random-forward} to gather tokens and then broadcasts the gathered tokens using network coded indexed-broadcast:\\[-0.15cm]

\boxtext{
\begin{tabbing}
else\= else\= else\= \kill
Algorithm {\tt greedy-forward}\\
\\[-0.15cm]
{\bf while} tokens remain to be broadcast\\
\>{\tt random-forward}\\
\>the identified node broadcasts up to $b^2/d$ tokens\\
\>\ \ \ (using the network coded indexed-broadcast)\\
\>remove all broadcast tokens from consideration
\end{tabbing}
}

\begin{theorem}\label{thm:alg1} 
  With high probability the {\tt greedy-forward} algorithm takes $O(nkd/b^2 + nb)$ time to solve
  the $k$-token dissemination problem. 
\end{theorem}
\begin{proof}
  Note that it is easy to check in $n$ rounds whether any node has any
  tokens to forward. Thus each iteration of the loop takes $O(n)$
  rounds.  Suppose that an iteration begins with $k'$
  tokens to be broadcast. Lemma~\ref{lem:randomforward} shows that
  at least $M=\sqrt{bk'/d}$ tokens will be gathered in one identified node
  by the {\tt random-forward} process. This node can then use the 
  network coded $k$-indexed-broadcast from Section~\ref{sec:NC} to broadcast 
  these tokens in $O(n)$ rounds. 

  Thus, so long as $M \ge b^2/d$, meaning $k' > b^3/d$, the
  algorithm will broadcast $b^2/d$ tokens every $O(n)$ rounds, which
  can happen at most $kd/b^2$ times.

  Once $k' \le b^3/d$, we no longer gather and broadcast the full
  $b^2/d$ tokens.  Instead, since the maximum number of tokens at
  a node after {\tt random-forward} is $\sqrt{\frac{bk'}{d}}$, we have the following
  recurrence for the number of $O(n$)-round phases $T(k')$ performed to
  transmit $k'$ items:\\
  \hspace*{6em} $T(k') \leq 1 + T(k' - \sqrt{\frac{k'b}{d}})$.
 
  We conclude that it requires $O(\sqrt{k'd/b})$
  phases to reduce the number of remaining items
  from $k'$ to $k'/2$.  Iterated halving yields a geometric series for
  the running time whose first term (when $k'=\Theta(b^3/d)$) dominates,
  giving $T(b^3/d) = O(b)$ phases of $O(n)$-time broadcasts which
  results in a running time of $O(nb)$ rounds in the end. Putting both
  parts together gives that the total time to collect all tokens is
  $O(nkd/b^2 + nb)$.
\end{proof}

Observe that this algorithm does not pay the extra $\log n$ factor
introduced by the naive indexed-broadcast algorithm.  Because all
tokens to be broadcast are gathered to a single node, indexing is
trivial. This {\tt greedy-forward} algorithms works well for small $b$, but for
very large $b \geq n^{1/3}$ the {\tt random-forward} routine is not 
able to gather $b^2/d$ tokens in one node efficiently. For this scenario
we have a different algorithm that avoids the additive $nb$-round term.\\[-0.15cm]

\boxtext{
\begin{tabbing}
else\= else\= else\= \kill
Algorithm {\tt priority-forward}\\
\\[-0.15cm]
Run {\tt greedy-forward} until no node gets $b^2/d$ tokens\\
{\bf while} tokens remain to be broadcast\\
\>Nodes group tokens into blocks of size $b/d$\\
\>Assign each block a random $O(\log n)$-bit priority\\
\>Index $\Theta(b)$ random blocks in $O(n)$ time\\
\>\>(using {\tt priority-forward} recursively (*))\\
\>Broadcast these blocks in $O(n)$ time\\
\>\>(using the network coded indexed broadcast)\\
\>remove all broadcast tokens from consideration
\end{tabbing}
}

\begin{lemma}
With high probability {\tt priority-forward} will terminate in $O((1+kd/b^2)\log n)$ iterations of its while loop.
\end{lemma}
\begin{proof}
The while loop starts when
  no node learns of more than $b^2/d$ tokens during {\tt     random-forward}. In this case we know from
  the proof of Lemma~\ref{lem:randomforward} that afterwards the   number of
  nodes $c_i$ that know about each token $i$ is
  $\Omega(\frac{n}{b})$ with high probability. Let $C=\sum c_i$.

  The algorithm divides the known tokens into blocks of size $b/d$ and   picks $b$ random blocks.  There are at most
  $C/(b/d)$ full blocks in total and at most one partially-full block   per node
  for a total of $n$ partially full blocks.  We consider two cases.

  If $C/(b/d)<n$ then there are at most $2n$ blocks in total.  Since   with high probability every token is in $\Omega(n/b)$
  blocks, one of these blocks is among the chosen $b$ with probability
  at least $(1-1/2b)^b=\Omega(1)$.  It follows that after $O(\log n)$   rounds involving less than $n$ full blocks, all tokens will be   chosen and
  disseminated with high probability.

  If $C/(b/d) > n$ then the number of blocks is at most $2C/(b/d)$.
  We argue in this case that $C$ decreases in expectation by a factor
  of $e^{-b^2/kd}$ in each iteration.  If this is true then after
  $kd(\log n)/b^2$ rounds the expected decrease is polynomial; since
  $C$ was polynomial to begin with its expected value will be
  polynomially small.  At this point the Markov bound indicates
  that $C=0$ with high probability.

  To show the expected decrease, note there are at most
  $2C/(b/d)$ blocks of which $c_i$ contain item $i$.  Thus, when a
  random block is chosen, item $i$ is in it with
  probability at least $c_i(b/d)/2C$.  So item $i$ fails to be   chosen with probability at most   $(1-bc_i/2Cd)^b<\exp(-(b^2/d)c_i/2C)$.
  If we let $c'_i=c_i$ for tokens not chosen, and $c'_i=0$ for tokens
  that are, we find $E[\sum c'_i] \le \sum
  c_i\exp(-(b^2/d)c_i/C) = C\sum \alpha_i \exp(-(b^2/d)\alpha_i)$   where
  $\alpha_i = c_i/C$ so $\sum \alpha_i = 1$.  Differentiating
  shows this sum is maximized when all $\alpha_i$ are set
  equal at $1/k$ (since there are at most $k$ distinct $\alpha_i$),
  yielding a value of $C\exp(-(b^2/kd))$.  It follows that the
  expected value of $\sum c_i$ decreases by a factor $e^{-b^2/kd}$ in
  each round.
\end{proof}

We have shown that a small number of iterations suffices but must asses the time to implement one iteration.  In particular, we must explain how line (*) in {\tt priority-forward} can be implemented.  To choose $b$ random blocks, we give each block a random $O(\log n)$ bit priority (so collisions are unlikely) and then identify and index the $b$ lowest priorities. Since block priorities have size $O(\log n)$, we can treat their identification as an indexing problem with $d=O(\log n)$. The naive indexing algorithm via flooding requires $O(n\log n)$ time to broadcast the $b$ lowest priority blocks ($b/\log n$ blocks every $O(n)$ rounds). This would lead to a runtime of $O(nkd(\log^2 n)/b^2 + n \log^2 n)$. We can reduce the running time by a $\log n$ factor with a more careful approach, which calls {\tt priority-forward} recursively to disseminate $\Theta(b)$ of the smallest size-$O(\log n)$ priorities in only $O(n)$ time on every iteration of the while loop. We defer the details to the full paper. We get the following for the performance of the {\tt priority-forward} algorithm:

\begin{theorem}\label{thm:random-spreading-extended}
  For $b \geq \log^3 n$, {\tt priority-forward} 
  solves $k$-token dissemination in $O(\frac{\log n}{b} \cdot \frac{nkd}{b} + n \log n)$ rounds with high probability.
\end{theorem}

\section{Exploiting $T$-stability}

In this section we consider more stable networks and show how to
design faster protocols in such a setting.

Kuhn et al. introduced the notion of \emph{$T$-interval connectivity}
to define more stable networks in which over every block of $T$ rounds
at least a spanning-subgraph is unchanging.  They give algorithms with
linear speedup in $T$ and matching lower bound for knowledge-based token-forwarding algorithms.  We work with our related but stronger requirement of
\emph{$T$-stability} which demands that the entire network changes only every $T$
steps.  Although the Kuhn et al. lower bound for token forwarding still holds in this model,
we give network-coding algorithms with a quadratic speedup in $T$. This $T^2$ speedup
comes from two ideas, each contributing a factor of $T$.  The first is that in a $T$-stable network a node can communicate to the same neighbor $T$ times, thus passing a message $T$ times as large.  This does cost a factor-$T$ slowdown in the time to send a message, but the results of section~\ref{sec:generalization} show that the communication rate increases as $T^2$.   Combining these factors nets a factor-$T$ overall improvement.  The second idea, drawn from Kuhn et al., is that in $T$ rounds pipelining enables a node to communicate its (enlarged) message to at least $T$ nearby nodes simultaneously, yielding a second factor-$T$ speedup.
We currently need to rely on the notion of
$T$-stability for this, but we speculate that $T$-interval
connectivity might suffice. The technique composes with the our technique exploiting larger message sizes 
from the previous section and leads to quadratic speed ups in $b$ and $T$ for most settings of these parameters.

As previously, we begin by describing an efficient indexed-broadcast algorithm
and then show how it can be used as a primitive for $k$-token dissemination.

Our indexed broadcast algorithm exploits $T$-stability to broadcast $bT$
blocks each containing $bT$ bits, for a total of $(bT)^2$ bits (or
$(bT)^2/d$ tokens), in $O((n+bT^2)\log n)$ rounds.  
As before, we use network coding, treating these blocks as vectors and
flooding random linear combinations of the vectors through the
network.  We do so by dividing the network, in each block of $T$
stable rounds, into \emph{patches} of size and diameter roughly $T$.
We then spread random linear combinations of the size-$bT$ blocks from
patch to patch, taking $O(T)$ rounds to spread to each new patch but
reaching $T$ nodes in the patch each time, so that $n$ rounds suffice
for all nodes to receive all necessary linear combinations.

\subsection{Patching the Graph}

Our first step is to partition the graph into connected \emph{patches}
of size $\Omega(D)$ and diameter $O(D)$. It helps to think of $D$ as approximately
$T$; Because computing the patching takes $D \log n$ time, we will set $D=O(T/\log n)$.
We will use these patches
for $O(T)$ rounds, during which they will remain static.  First, we argue
that such patches exist.  Let $G^D$ be the $D^{th}$ power of the (unchanging)
connectivity graph---in other words, connect every node to any node
within distance $D$.  Consider a \emph{maximal independent set} $S$ in
$G^D$.  If every vertex in $G$ is assigned to the closest vertex in
$S$, we get patches that satisfy our criteria:
\begin{enumerate}
\item Consider a shortest path tree on the vertices assigned to vertex $u \in
  S$.  If $v$ is assigned to $u$, then so are the ancestors of $v$ in the
  shortest paths tree.  Thus, the shortest path tree connects the patch.
\item Because of the maximality of $S$, every vertex is adjacent in $G^D$ to a vertex in
  $S$, since otherwise such a vertex could be added to $S$. In other words, any
  vertex is within distance $D$ of $S$.  It follows that the depth of
  each shortest paths tree, which bounds the (half of the) diameter, is at most $D$
\item Also by definition, no two vertices in $S$ are adjacent in
  $G^D$---in other words, their distance in $G$ exceeds $D$.  Thus,
  any vertex within distance $D/2$ of $u\in S$ is assigned to $u$.  It
  follows that every patch has at least $D/2$ vertices.
\end{enumerate}

It remains to \emph{construct} such a maximal independent
set.  Luby's maximal independent set
permutation algorithm \cite{luby1985simple} can be easily adapted to run in
our model.  In Luby's permutation algorithm, vertices talk to their ``neighbors''.
Since we are computing in the powered graph $G'$, we need vertices to
talk to other vertices at distance $D$ over long communication paths.
We have $T$ time, but different communication paths may overlap,
causing congestion.

Fortunately, this is not a significant problem.  The core step of
Luby's algorithm assigns every vertex a random priority, then adds to
the MIS any vertex whose priority is higher than all its neighbors and
``deactivates'' all its neighbors.  Thus, nodes need only learn the
maximum priority of any neighbor and notify neighbors of their
deactivation. We can simulate the procedure.  Nodes can find the
highest priority within distance $D$ by flooding the highest priority they
hear for $D$ rounds.  If a node hears no higher priority than its own,
then it knows it is in the MIS and can broadcast a ``deactivation''
message to all nodes within distance $D$ of itself. Luby's algorithm
runs in $O(\log n)$ time, which translates to $O(D \log n)$ here. We thus 
choose $D = O(T/\log n)$.

\subsection{$T$-Stable Indexed-Broadcast}

Given our patches of the required size and diameter, we use network coding to
distribute vectors of $bT$ bits.  In a
particular sequence of $O(T)$ rounds, after having computed the
patches for this sequence, we do the following:
\begin{enumerate}
\item \emph{share:} All nodes in a patch jointly share a random linear combination of 
  the vectors in the \emph{union of all} their received messages, each adding the result to its own set of received messages
\item \emph{pass:} Each node broadcasts its patch's agreed random sum vector to its neighbors
\item \emph{share:} The first sharing phase is repeated, including the
  new vectors just received from neighbors.
\end{enumerate}

\subsubsection{Implementation}

We show how to implement all the required steps in $O(T)$ rounds.  The
middle pass step is trivial: each node breaks its size-$bT$ vector
into $T$ components of size $b$ and transmits one component in each
round.  Neighbors receive and reassemble all components.

Less trivial is the share step. We show how all the nodes in a given
patch can compute a random sum of all the size-$bT$ vectors in
all their received messages.

For this we use the vertices in the maximal independent set $S$ as leaders
and assume that each patch has agreed on a (shortest
path) tree rooted at the leader; each node knows its depth and its parent and children.
This can be done by letting the leader send out an incrementing broadcast for $O(D)$ rounds.
The time when this broadcast reaches a node tells it its depth and the (lowest ID)
node that the broadcast was received from is the ``parent''. 

Now we want to compute a random linear combination of the \emph{union} of all the vectors in all the nodes
of the patch.  First, each node just computes a random sum of its own
vectors.  It remains to sum these sums.  This would be easy
if the vectors had dimension $b$---we would pass them up from children
to parents, summing as we went, so that each node only passed up one
vector.  Since their dimension is $bT$ we pipeline.  Each node
breaks its length-$bT$ vector $(v_1,\ldots,v_{bT})$ into $T$
length-$b$ vectors $w_i=(v_{iT},v_{iT+1},\ldots,v_{(i+1)T-1})$.  At
step $s$ of this phase, any node at depth $j$ will have the cumulative
sum of all the $w_{s+j-T}$ components of the vectors from its
descendants.  It broadcasts this sum to its parent, and at the same
time receives from its children their own cumulative $w_{s+(j+1)-T}$
sums.  The receiving node adds these children's' sums to its own
$w_{s+j+1+T}$ component, producing the cumulative $w_{(s+1)+j-T}$
component sum that it needs to transmit the next round.  After $T + D < 2T$
time steps, the root will have received cumulative sums of all the
$w_i$ vectors from its children and added them, yielding the sum of
all the vectors, which is a random sum of all the basis vectors.

This random sum, a single size $bT$-vector, is now distributed by the
leader to all nodes in the patch via the obvious pipelined broadcast.

\subsubsection{Analysis}

We now analyze the share-pass-share algorithm outlined above.  As
before, we show that any vector $\mu$ that is \emph{sensed}
by (not perpendicular to the basis of) some node at the start is
quickly sensed by all vectors.  

\begin{lemma}\label{lem:stabilityNC}
With high probability the patch-sharing network coding algorithm solves the $bT$-indexed-broadcast problem in a $T$-stable dynamic network with tokens of size $bT$ in $O((n+bT^2) \log n)$ rounds using messages of size $O(b)$.
\end{lemma}
This is close to the best achievable time.  The $n$
term follows from the network's possible $n$ diameter.  The $bT^2$
term follows from information theory: the $b^2T^2$ bits we aim to
transmit may be at a single node that broadcasts only $b$ bits per
round, implying $bT^2$ rounds will be necessary for that node to
broadcast its information.  
\begin{proof}
To simplify our proof we assume that $bT^2\le n$ and prove an $O(n\log n)$ bound.  For if $bT^2 \ge n$, we
can run our algorithm for $t<T$ such that $bt^2=n$ and distribute
$b^2t^2$ bits in $O(n\log n)$ rounds; repeating $(T/t)^2$
times will distribute all the bits in $(T/t)^2n\log n=(T^2n/t^2)\log n=T^2b\log n$ rounds.

Since we are operating on size-$bT$ messages we can allow $bT$ tokens of size $bT$, each with a $\log q = O(1)$ size coefficient.   We consider share-pass-share ``meta rounds'' of length $T$ where our patches are fixed, and show that $O(n/D)$ of these meta rounds suffice to disseminate all the tokens, for a total of $O(T*(n/D))=O(n\log n)$ rounds.
For a given meta round we consider two cases.  The first is where there is some patch that
contains no node sensing $\mu$.  In this case, the connectivity
assumption implies that a node $u$ in some such patch is adjacent to
some node $v$ in a patch containing a node that does sense $\mu$.  In
the first share step $v$ receives a random linear combination of
the vectors in its patch; since some node in the patch senses $\mu$,
with probability $1-1/q$ node $v$ will sense $\mu$ after the first 
sharing phase. In this case $v$ transmits the same random linear
combination to $u$ in the pass phase and $u$ will sense $\mu$ as well.
If so, the final share step will deliver to all nodes in $u$'s patch a linear
combination not perpendicular to $\mu$ with probability $1-1/q$.
Combining these arguments, we find that with probability $(1-1/q)^2$,
the $\Omega(D)$ nodes in $u$'s patch, which previously did not sense
$\mu$, will now do so.

The second case is where every patch contains a node that senses $\mu$.
In this case every node has a $1-1/q$ chance of sensing $\mu$ after
the first share step.  The expected number of nodes that do \emph{not}
sense $\mu$ thus shrinks by a $1/q<1/2$ factor.  The Markov bound
shows that it thus shrinks by a factor $2/3$ with constant
probability. 

We now combine the two cases.  If case 1 holds declare a success if
$\Omega(D)$ new nodes sense $\mu$; if case 2 holds declare a success
if the number of nodes that do not sense $\mu$ shrinks by $2/3$.
There can be only $O(n/D)$ successes of case 1 and $O(\log n)$ successes of
case 2 before all nodes sense $\mu$.  A Chernoff bound shows that
within $\Omega(n/D)$ occurrences of case 1 the probability that we fail to
observe $O(n/D)$ successes is $e^{-\Omega(n/D)}$.  Similarly, the probability
of less than $\log n$ successes in $\Omega(n/D)$ occurrences of case 2 is
$e^{-\Omega(n/D)}$ (this follows from the fact that $T^2<n$, meaning $n/T > \log n$).   

Finally, we apply a union bound on the above argument over all the
$2^{bT}$ distinct vectors of size $T$.  The
probability \emph{any} such vector fails to be sensed in $\Omega(n/D)$
phases is then at most $2^{bT}e^{-\Omega(n/D)}$ which is negligible given
our assumption that $bT^2 \le n$. Thus in $\Theta(n/D)$ phases each with a running time of $O(T)$, totaling $O(n\log n)$ time, all nodes 
sense all vectors and can decode all tokens.
\end{proof}

This algorithm can be derandomized using the arguments developed in Section~\ref{sec:det}
and replacing Luby's randomized MIS algorithm by the 
deterministic distributed MIS algorithm in \cite{panconesi1992improved} with a running time of $MIS(n) = 2^{O(\sqrt{\log n})}$.
The larger $k^2 \log n$ coefficient overhead still allows for $\sqrt{bT/\log n}$ tokens of size $O(bT)$
being code together for a vector size of $O(bT)$. This leads to the following Lemma:

\begin{lemma}\label{lem:stabilityNCdet} \
The deterministic patch-sharing algorithm solves the $\sqrt{bT/\log n}$-indexed-broadcast problem with tokens of size $bT$ in a $T$-stable dynamic network in $O((n + \sqrt{bT}T) \cdot MIS(n))$ rounds using messages of size $O(b)$.
\end{lemma}

\subsection{$T$-Stable Token Dissemination}\label{sec:stable-token-dissemination}

We have given an $O(n \log n)$-time algorithm for indexed broadcast of
$bT$ vectors of $bT$ bits.  Applying the same reduction(s) as before,
we might hope to achieve a $k$-token dissemination algorithm with running
time $O(n \log n \frac{kd}{(bT)^2})$.  This can be achieved for most values of
$k,b$ and $T$. The key, as before,
is \emph{gathering} tokens we wish to broadcast as large blocks/tokens.  Since the blocks used
with $T$-stability are larger, gathering is harder. In particular:
\begin{itemize}
\item Using {\tt greedy-forward} to gather tokens yields an algorithm
  with running time $O(\frac{\log n}{bT^2} \cdot \frac{nkd}{b} \  + \ nbT^2\log n)$
\item Using {\tt priority-forward} to gather tokens yields an
  algorithm with running time $O(\frac{\log^2 n}{bT^2} \cdot \frac{nkd}{b} \  + \ nT\log^2 n)$.
\end{itemize}
The second algorithm is near-optimal unless $T$ is very large.  In
this case there is an alternative gathering algorithm we can apply:
create the patches of our patch algorithm, then use pipelining to
gather together the tokens in a patch to blocks of size at most $bT$
at a single node (or, if there is more than one block, at multiple nodes)
of that patch.  This produces $O(n/D + kd/bT) = O(n \log n/T)$ blocks of 
size at most $O(bT)$ which can be much smaller than $k$. In phases
of $O(n \log n)$ rounds we then index $bT$ of these blocks or tokens
using pipelined flooding and broadcast them out using the network coded 
indexed-broadcast algorithm. This leads to an $O(\frac{1}{bT} \cdot \min\{k,\frac{n \log n}{T}\} + 1) \cdot n \log n$
round algorithm for $k$-token dissemination. This completes the results 
stated in Theorem~\ref{thm:stable}. 

For deterministic algorithms gathering is much harder. Considering the limitations
of token-forwarding, it seems unlikely that the gathering methods
that are based on the {\tt random-forward} primitive
can be derandomized. Nevertheless, we can
make the last gathering method deterministic by using the 
deterministic MIS algorithm from \cite{panconesi1992improved} once more.
This, together with the deterministic indexed-broadcast algorithm
from Lemma~\ref{lem:stabilityNCdet}, leads to an $O(n/D + kd/bT) / \sqrt{bT/\log n} \cdot O(n \cdot MIS(n)) = 
O(\frac{MIS(n)^2 \sqrt{\log n}}{\sqrt{bT}T} \cdot n^2 + n \cdot MIS(n))$ algorithm as stated in 
Theorem~\ref{thm:stabledeterministic}; here $MIS(n)$ is the time needed
to compute a maximum independent set in an $n$ node graph. 

Allowing centralized algorithms on the other hand alleviates many of these problems: indices can be assigned trivially and 
the coefficient overhead can be ignored since it is easy to infer the coefficients
from knowing the past topologies. This allows a randomized centralized algorithm
to distribute $n$ blocks of size $O(b)$ in $O(n)$ time and leads to a linear time
algorithm for the $k$-token dissemination problem as stated in Corollary \ref{cor:global}.
To obtain deterministic centralized algorithm we have to be 
more careful: A deterministic centralized algorithm that codes together $k$ tokens requires according to 
Corollary \ref{col:NCdeterministic} a field size $q = n^{k}$. In order
to describe one symbol in the $bT$-bit size blocks, that are used in the algorithm
developed in this section, at most $k = bT/\log n$ blocks of size $bT$ can be coded together.
We also note, that with central control the MIS computation becomes local and thus trivial. 
Putting all this together and using the third (deterministic) gathering technique leads to
the results stated in Corollary \ref{cor:global}.

\section{Conclusion}

We have applied \emph{network coding}
to distributed computing in dynamic networks.
We provided faster algorithms for distributed information dissemination
which, in several cases, work provably better than \emph{any} non-coding algorithm. 

Message size is an important parameter that 
was not fully accounted for in previous work: while extremely small
(logarithmic size) messages are a standard assumption in distributed computing,
prior work on network coding assumed exponentially larger, linear size 
messages. We mediate between these two assumptions using an explicit message size
and show that, contrary to the natural assumption that broadcast
should scale linearly with the message size, it can be made to scale \emph{quadratically} using
network coding. 

We also explore the range between fully dynamic and fully static networks, showing that in $T$-stable networks dissemination can be sped up by a factor of $T^2$ using network coding.  In contrast, the Kuhn et al. lower bound apply to such $T$-stable networks and show that knowledge-based token-forwarding algorithm can only offer a factor-$T$ speedup. Improving our patch-sharing algorithms to avoid the computation of a
maximum independent set and making them applicable to the
$T$-interval-connectivity model remains an interesting question.  So
far we can achieve this goal only if the topologies chosen by the
adversary are highly non-expanding.

Many of our algorithmic ideas can be extended beyond
the always-connected dynamic networks discussed in this paper to other
network and communication models~\cite{AnalyzingNC}. The same is true
for our results on omniscient adversaries or (non-uniformly) deterministic 
algorithms. 

We have shown that network coding outperforms token forwarding, but it is not clear whether we have made best-possible use of this technique.  Conceivably network coding can yield even better performance.  Unlike for token forwarding, there are no non-trivial lower bounds for general or network-coding based algorithms for $n$-token dissemination in the dynamic network model. Closing this gap is an intriguing open question.

%\paragraph{Acknowledgments}
\section*{Acknowledgments} We thanks Nancy Lynch and Rotem Oshman for introducing us to the dynamic network model.
We thank Muriel M\'{e}dard and Lizhong Zheng for interesting discussions. Lastly, 
we thank the anonymous reviewers for helpful comments.

\end{document}